\documentclass{article}
\usepackage{amsmath}
\usepackage{amssymb}
\usepackage{amsfonts}
\begin{document}
\newtheorem{thm}{Theorem}
\newtheorem{lem}[thm]{Lemma}
\newtheorem{prop}[thm]{Proposition}
\newtheorem{cor}[thm]{Corollary}
\newtheorem{assum}[thm]{Assumption}
\newtheorem{rem}[thm]{Remark}
\newtheorem{notation}[thm]{Notation}
\newtheorem{defn}[thm]{Definition}
\newtheorem{clm}[thm]{Claim}
\newcommand{\unit}{\mathbb I}
\newcommand{\ali}[1]{{\mathfrak A}_{[ #1 ,\infty)}}
\newcommand{\alm}[1]{{\mathfrak A}_{(-\infty, #1 ]}}
\newcommand{\nn}[1]{\lV #1 \rV}
\newcommand{\br}{{\mathbb R}}
\newcommand{\dm}{{\rm dom}\mu}
\newcommand{\lb}{l_{\bb}(n,n_0,k_R,k_L,\lal,\bbD,\bbG,Y)}
\newcommand{\Ad}{\mathop{\mathrm{Ad}}\nolimits}
\newcommand{\Proj}{\mathop{\mathrm{Proj}}\nolimits}
\newcommand{\sgn}{\mathop{\mathrm{sgn}}\nolimits}
\newcommand{\id}{\mathop{\mathrm{id}}\nolimits}
\newcommand{\RRe}{\mathop{\mathrm{Re}}\nolimits}
\newcommand{\RIm}{\mathop{\mathrm{Im}}\nolimits}
\newcommand{\Wo}{\mathop{\mathrm{Wo}}\nolimits}
\newcommand{\Prim}{\mathop{\mathrm{Prim}_1}\nolimits}
\newcommand{\Primz}{\mathop{\mathrm{Prim}}\nolimits}
\newcommand{\ClassA}{\mathop{\mathrm{ClassA}}\nolimits}
\newcommand{\Class}{\mathop{\mathrm{Class}}\nolimits}
\def\qed{{\unskip\nobreak\hfil\penalty50
\hskip2em\hbox{}\nobreak\hfil$\square$
\parfillskip=0pt \finalhyphendemerits=0\par}\medskip}
\def\proof{\trivlist \item[\hskip \labelsep{\bf Proof.\ }]}
\def\endproof{\null\hfill\qed\endtrivlist\noindent}
\def\proofof[#1]{\trivlist \item[\hskip \labelsep{\bf Proof of #1.\ }]}
\def\endproofof{\null\hfill\qed\endtrivlist\noindent}
\newcommand{\caA}{{\mathcal A}}
\newcommand{\caB}{{\mathcal B}}
\newcommand{\caC}{{\mathcal C}}
\newcommand{\caD}{{\mathcal D}}
\newcommand{\caE}{{\mathcal E}}
\newcommand{\caF}{{\mathcal F}}
\newcommand{\caG}{{\mathcal G}}
\newcommand{\caH}{{\mathcal H}}
\newcommand{\caI}{{\mathcal I}}
\newcommand{\caJ}{{\mathcal J}}
\newcommand{\caK}{{\mathcal K}}
\newcommand{\caL}{{\mathcal L}}
\newcommand{\caM}{{\mathcal M}}
\newcommand{\caN}{{\mathcal N}}
\newcommand{\caO}{{\mathcal O}}
\newcommand{\caP}{{\mathcal P}}
\newcommand{\caQ}{{\mathcal Q}}
\newcommand{\caR}{{\mathcal R}}
\newcommand{\caS}{{\mathcal S}}
\newcommand{\caT}{{\mathcal T}}
\newcommand{\caU}{{\mathcal U}}
\newcommand{\caV}{{\mathcal V}}
\newcommand{\caW}{{\mathcal W}}
\newcommand{\caX}{{\mathcal X}}
\newcommand{\caY}{{\mathcal Y}}
\newcommand{\caZ}{{\mathcal Z}}
\newcommand{\bbA}{{\mathbb A}}
\newcommand{\bbB}{{\mathbb B}}
\newcommand{\bbC}{{\mathbb C}}
\newcommand{\bbD}{{\mathbb D}}
\newcommand{\bbE}{{\mathbb E}}
\newcommand{\bbF}{{\mathbb F}}
\newcommand{\bbG}{{\mathbb G}}
\newcommand{\bbH}{{\mathbb H}}
\newcommand{\bbI}{{\mathbb I}}
\newcommand{\bbJ}{{\mathbb J}}
\newcommand{\bbK}{{\mathbb K}}
\newcommand{\bbL}{{\mathbb L}}
\newcommand{\bbM}{{\mathbb M}}
\newcommand{\bbN}{{\mathbb N}}
\newcommand{\bbO}{{\mathbb O}}
\newcommand{\bbP}{{\mathbb P}}
\newcommand{\bbQ}{{\mathbb Q}}
\newcommand{\bbR}{{\mathbb R}}
\newcommand{\bbS}{{\mathbb S}}
\newcommand{\bbT}{{\mathbb T}}
\newcommand{\bbU}{{\mathbb U}}
\newcommand{\bbV}{{\mathbb V}}
\newcommand{\bbW}{{\mathbb W}}
\newcommand{\bbX}{{\mathbb X}}
\newcommand{\bbY}{{\mathbb Y}}
\newcommand{\bbZ}{{\mathbb Z}}
\newcommand{\caSP}{{\mathcal SP}}
\newcommand{\str}{^*}
\newcommand{\lv}{\left \vert}
\newcommand{\rv}{\right \vert}
\newcommand{\lV}{\left \Vert}
\newcommand{\rV}{\right \Vert}
\newcommand{\la}{\left \langle}
\newcommand{\ra}{\right \rangle}
\newcommand{\ltm}{\left \{}
\newcommand{\rtm}{\right \}}
\newcommand{\lcm}{\left [}
\newcommand{\rcm}{\right ]}
\newcommand{\ket}[1]{\lv #1 \ra}
\newcommand{\bra}[1]{\la #1 \rv}
\newcommand{\lmk}{\left (}
\newcommand{\rmk}{\right )}
\newcommand{\al}{{\mathcal A}}
\newcommand{\md}{M_d({\mathbb C})}
\newcommand{\Tr}{\mathop{\mathrm{Tr}}\nolimits}
\newcommand{\Ran}{\mathop{\mathrm{Ran}}\nolimits}
\newcommand{\Ker}{\mathop{\mathrm{Ker}}\nolimits}
\newcommand{\spn}{\mathop{\mathrm{span}}\nolimits}
\newcommand{\Mat}{\mathop{\mathrm{M}}\nolimits}
\newcommand{\UT}{\mathop{\mathrm{UT}}\nolimits}
\newcommand{\DT}{\mathop{\mathrm{DT}}\nolimits}
\newcommand{\GL}{\mathop{\mathrm{GL}}\nolimits}
\newcommand{\spa}{\mathop{\mathrm{span}}\nolimits}
\newcommand{\supp}{\mathop{\mathrm{supp}}\nolimits}
\newcommand{\rank}{\mathop{\mathrm{rank}}\nolimits}
\newcommand{\idd}{\mathop{\mathrm{id}}\nolimits}
\newcommand{\ran}{\mathop{\mathrm{Ran}}\nolimits}
\newcommand{\dr}{ \mathop{\mathrm{d}_{{\mathbb R}^k}}\nolimits} 
\newcommand{\dc}{ \mathop{\mathrm{d}_{\cc}}\nolimits} \newcommand{\drr}{ \mathop{\mathrm{d}_{\rr}}\nolimits} 
\newcommand{\zin}{\mathbb{Z}}
\newcommand{\rr}{\mathbb{R}}
\newcommand{\cc}{\mathbb{C}}
\newcommand{\ww}{\mathbb{W}}
\newcommand{\nan}{\mathbb{N}}\newcommand{\bb}{\mathbb{B}}
\newcommand{\aaa}{\mathbb{A}}\newcommand{\ee}{\mathbb{E}}
\newcommand{\pp}{\mathbb{P}}
\newcommand{\wks}{\mathop{\mathrm{wk^*-}}\nolimits}
\newcommand{\he}{\hat {\mathbb E}}
\newcommand{\ikn}{{\caI}_{k,n}}
\newcommand{\mk}{{\Mat_k}}
\newcommand{\mnz}{\Mat_{n_0}}
\newcommand{\mn}{\Mat_{n}}
\newcommand{\mkk}{\Mat_{k_R+k_L+1}}
\newcommand{\mnzk}{\mnz\otimes \mkk}
\newcommand{\hbb}{H^{k,\bb}_{m,p,q}}
\newcommand{\gb}[1]{\Gamma^{(R)}_{#1,\bb}}
\newcommand{\cgv}[1]{\caG_{#1,\vv}}
\newcommand{\gv}[1]{\Gamma^{(R)}_{#1,\vv}}
\newcommand{\gvt}[1]{\Gamma^{(R)}_{#1,\vv(t)}}
\newcommand{\gbt}[1]{\Gamma^{(R)}_{#1,\bb(t)}}
\newcommand{\cgb}[1]{\caG_{#1,\bb}}
\newcommand{\cgbt}[1]{\caG_{#1,\bb(t)}}
\newcommand{\gvp}[1]{G_{#1,\vv}}
\newcommand{\gbp}[1]{G_{#1,\bb}}
\newcommand{\gbpt}[1]{G_{#1,\bb(t)}}
\newcommand{\Pbm}[1]{\Phi_{#1,\bb}}
\newcommand{\Pvm}[1]{\Phi_{#1,\bb}}
\newcommand{\mb}{m_{\bb}}
\newcommand{\E}[1]{\widehat{\mathbb{E}}^{(#1)}}
\newcommand{\lal}{{\boldsymbol\lambda}}
\newcommand{\mm}{{\boldsymbol\mu}}
\newcommand{\oo}{{\boldsymbol\omega}}
\newcommand{\vv}{{\boldsymbol v}}
\newcommand{\bbm}{{\boldsymbol m}}
\newcommand{\kl}[1]{{\mathcal K}_{#1}}
\newcommand{\wb}[1]{\widehat{B_{\mu^{(#1)}}}}
\newcommand{\ws}[1]{\widehat{\psi_{\mu^{(#1)}}}}
\newcommand{\wsn}[1]{\widehat{\psi_{\nu^{(#1)}}}}
\newcommand{\wv}[1]{\widehat{v_{\mu^{(#1)}}}}
\newcommand{\wbn}[1]{\widehat{B_{\nu^{(#1)}}}}
\newcommand{\wo}[1]{\widehat{\omega_{\mu^{(#1)}}}}
\newcommand{\dist}{\dc}
\newcommand{\hpu}{\hat P^{(n_0,k_R,k_L)}_R}
\newcommand{\hpd}{\hat P^{(n_0,k_R,k_L)}_L}
\newcommand{\pu}{ P^{(k_R,k_L)}_R}
\newcommand{\pd}{ P^{(k_R,k_L)}_L}
\newcommand{\puuz}{P_{R}^{(n_0-1,n_0-1)}\otimes P^{(k_R,k_L)}_R}
\newcommand{\pddz}{P_{L}^{(n_0-1,n_0-1)}\otimes P^{(k_R,k_L)}_L}
\newcommand{\puu}{\tilde P_R}
\newcommand{\pdd}{\tilde P_L}
\newcommand{\qu}[1]{ Q^{(k_R,k_L)}_{R, #1}}
\newcommand{\qd}[1]{ Q^{(k_R,k_L)}_{L,#1}}
\newcommand{\hqu}[1]{ \hat Q^{(n_0,k_R,k_L)}_{R, #1}}
\newcommand{\hqd}[1]{ \hat Q^{(n_0,k_R,k_L)}_{L,#1}}
\newcommand{\ek}[1] {e^{(k)}_{#1}}
\newcommand{\eij}[1] {E^{(k_R,k_L)}_{#1}}
\newcommand{\eijz}[1] {E^{(n_0-1,n_0-1)}_{#1}}
\newcommand{\heij}[1] {\hat E^{(k_R,k_L)}_{#1}}
\newcommand{\cn}{\mathop{\mathrm{CN}(n_0,k_R,k_L)}\nolimits}
\newcommand{\ghd}[1]{\mathop{\mathrm{GHL}(#1,n_0,k_R,k_L,\bbG)}\nolimits}
\newcommand{\ghu}[1]{\mathop{\mathrm{GHR}(#1,n_0,k_R,k_L,\bbD)}\nolimits}
\newcommand{\ghdb}[1]{\mathop{\mathrm{GHL}(#1,n_0,k_R,k_L,\bbG)}\nolimits}
\newcommand{\ghub}[1]{\mathop{\mathrm{GHR}(#1,n_0,k_R,k_L,\bbD)}\nolimits}
\newcommand{\hfu}[1]{{\mathfrak H}_{#1}^R}
\newcommand{\hfd}[1]{{\mathfrak H}_{#1}^L}
\newcommand{\hfui}[1]{{\mathfrak H}_{#1,1}^R}
\newcommand{\hfdi}[1]{{\mathfrak H}_{#1,1}^L}
\newcommand{\hfuz}[1]{{\mathfrak H}_{#1,0}^R}
\newcommand{\hfdz}[1]{{\mathfrak H}_{#1,0}^L}
\newcommand{\CN}{\overline{\hpd}\lmk\mnzk \rmk\overline{\hpu}}
\newcommand{\cnz}[1] {\chi_{#1}^{(n_0)}}
\newcommand{\eu}{\eta_{R}^{(k_R,k_L)}}
\newcommand{\ezu}{\eta_{R}^{(n_0-1,n_0-1)}}
\newcommand{\ed}{\eta_{L}^{(k_R,k_L)}}
\newcommand{\ezd}{\eta_{L}^{(n_0-1,n_0-1)}}
\newcommand{\fii}[1]{f_{#1}^{(k_R,k_L)}}
\newcommand{\fiir}[1]{f_{#1}^{(k_R,0)}}
\newcommand{\fiil}[1]{f_{#1}^{(0,k_L)}}
\newcommand{\fiz}[1]{f_{#1}^{(n_0-1,n_0-1)}}
\newcommand{\zeij}[1] {e_{#1}^{(n_0)}}
\newcommand{\CL}{\ClassA}
\newcommand{\CLn}{\Class_2(n,n_0,k_R,k_L)}
\newcommand{\braket}[2]{\left\langle#1,#2\right\rangle}
\newcommand{\abs}[1]{\left\vert#1\right\vert}
\newtheorem{nota}{Notation}[section]
\def\qed{{\unskip\nobreak\hfil\penalty50
\hskip2em\hbox{}\nobreak\hfil$\square$
\parfillskip=0pt \finalhyphendemerits=0\par}\medskip}
\def\proof{\trivlist \item[\hskip \labelsep{\bf Proof.\ }]}
\def\endproof{\null\hfill\qed\endtrivlist\noindent}
\def\proofof[#1]{\trivlist \item[\hskip \labelsep{\bf Proof of #1.\ }]}
\def\endproofof{\null\hfill\qed\endtrivlist\noindent}

\title{$C^1$-Classification of gapped parent Hamiltonians of quantum spin chains with local symmetry}
\author{
{\sc Yoshiko Ogata}\footnote{Supported in part by
the Grants-in-Aid for
Scientific Research, JSPS.}\\
{\small Graduate School of Mathematical Sciences}\\
{\small The University of Tokyo, Komaba, Tokyo, 153-8914, Japan}
}

\maketitle
\begin{abstract}
We consider the family of gapped Hamiltonians introduced in~\cite{Fannes:1992vq, Nachtergaele:1996vc} on
the quantum spin chain $\bigotimes_{\bbZ}\mn$, {\it with local symmetry} given by a group $G$.
The $G$-symmetric gapped Hamiltonians are given by triples $(k,u,V)$, where $u$ is a projective unitary representation of $G$ on a finite dimensional space $\cc^k$,
and $V$ is an isometry from $\bbC^k$ to $\cc^n\otimes \cc^k$.
We show that Hamiltonians $H_0,H_1$, given by the triples 
 $(k^{(0)}, u^{(0)}, V^{(0)})$ and $(k^{(1)}, u^{(1)}, V^{(1)})$ are $C^1$-equivalent if
the projective representations $u^{(0)}$ and  $u^{(1)}$ are unitary equivalent.
\end{abstract}

\section{Introduction}
Recently, gapped ground state phases attract a lot of attentions \cite{Chen:2010gb}\cite {Chen:2011iq}\cite {Schuch:2011ve}\cite{Bachmann:2011kw} 
\cite{Chen13}
\cite{bo}\cite{SW}.
In quantum spin systems, we say two gapped Hamiltonians are equivalent if and only if they are connected by a continuous path of uniformly
gapped Hamiltonians.  When we further require the path to be $C^1$, we call it $C^1$-classification. It is known that the ground state structure is an invariant of
the $C^1$-classification~\cite{Bachmann:2011kw}. The ultimate goal should be classifying all the gapped Hamiltonians in the world.
When we impose symmetry, the classification problem raises different mathematical question \cite{Schuch:2011ve}\cite{Chen13}.
Two Hamiltonians
which are equivalent in the classification without symmetry may be no longer in the same class if we consider the classification {\it with} symmetry. 

The general framework of classification with symmetry was considered in \cite{bns}. However, it is in general a hard problem to guarantee the existence of the spectral gap along the path rigorously.
As a result, examples of gapped Hamiltonians are quite limited in quantum spin systems whose spatial dimensions are larger than one.
However, for one dimensional systems, there is a recipe of gapped Hamiltonians \cite{Fannes:1992vq}.
We completely classified Hamiltonians given by this recipe, without breaking translation invariance in \cite{bo}. In \cite{Fannes:1992vq},
the existence of the spectral gap is guaranteed by
the primitivity of the associated completely positive map. 
We carried out the classification, by showing that the space of primitive maps, with an upper bound on Kraus rank, is connected, in \cite{bo}.
Recently, an alternative proof  of this fact was introduced in \cite{SW}. Both proof \cite{bo}, and \cite{SW} have their own advantages.
The proof in \cite{SW} is simpler and we don't need to care about the detailed structure, i.e., we don't need to know {\it how the primitivity is guaranteed} in details. Because of this simplicity, this proof can provide the analiticity of the path.
On the other hand, we find in \cite{Ogata1} and \cite{Ogata2} that the argument in \cite{bo} can be extended to the non-pimitive maps. It also has an advantage in constructing examples, as we know how the primitivity is guaranteed concretely.

In this paper, we consider the classification of the class of gapped Hamiltonians, given by the recipe of \cite{Fannes:1992vq}, {\it with} symmetry. The $G$-symmetric gapped Hamiltonians are given by triples $(k,u,V)$, where $u$ is a projective unitary representation of $G$ on a finite dimensional space $\cc^k$,
and $V$ is an isometry from $\bbC^k$ to $\cc^n\otimes \cc^k$.
This problem is considered in ~\cite{Schuch:2011ve}. The difference here is that we  do not break translation invariance. Mathematically, our task is to guarantee the primitivity of the associated CP-maps, {\it and} the symmetry. 
As the recipe in \cite{Fannes:1992vq} cares only about the primitive maps, we can use the method in \cite{SW}. About this problem, \cite{SW} is much stronger method than the one in \cite{bo}.
 
Now let us state our result in detail. See Appendix \ref{sec:nota} for notations.
For $\bbN\ni n\geq 2$, let $\caA$ be the finite dimensional C*-algebra $\caA = \Mat_n$, the algebra of $n\times n$ matrices. 
Throughout this article, this $n$ is fixed as the dimension of the spin under consideration, and we fix an orthonormal basis $\{\psi_\mu\}_{\mu=1}^n$ of $\cc^n$.
We denote the set of all finite subsets in ${\bbZ}$ by ${\mathfrak S}$.
The number of elements in a finite set $\Lambda\subset {\bbZ}$ is denoted by
$|\Lambda|$.
When we talk about intervals in $\bbZ$, $[a,b]$ for $a\le b$,
means the interval in $\bbZ$, i.e., $[a,b]\cap \bbZ$.
We denote the set of all finite intervals 
by ${\mathfrak I}$.
For each $z\in\bbZ$, we let $\caA_{\{x\}}$ be an isomorphic copy of $\caA$ and for any finite subset $\Lambda\subset\bbZ$, $\caA_{\Lambda} = \otimes_{x\in\Lambda}\caA_{\{x\}}$ is the local algebra of observables. 
For finite $\Lambda$, the algebra $\caA_{\Lambda} $ can be regarded as the set of all bounded operators acting on
a Hilbert space $\otimes_{x\in\Lambda}{\bbC}^n$.
We use this identification freely.
If $\Lambda_1\subset\Lambda_2$, the algebra $\caA_{\Lambda_1}$ is naturally embedded in $\caA_{\Lambda_2}$ by tensoring its elements with the identity. Finally,
the algebra $\caA_{\bbZ}$ 
is given as the inductive limit of the algebras $\caA_{\Lambda}$ with $\Lambda\in{\mathfrak S}$.
We denote the set of local observables by $\caA_{\bbZ}^{\rm loc}=\bigcup_{\Lambda\in{\mathfrak S}}\caA_{\Lambda}
$.
For any $x\in\bbZ$, let $\tau_x$ be the shift operator by $x$ on $\caA_\bbZ$. 

The local symmetry is introduced as follows.
Let $G$ be a group. Let $U:G\to \caU_n$ be a unitary representation of $G$ on $\bbC^n$.
Here, $\caU_n$ denotes the set of all unitary matrices on $\bbC^n$.
By $\beta_U$, we denote the product action of $G$ on $\caA_{\bbZ}$  induced by $U$, i.e., 
\begin{align}\label{eq:pa}
\beta_{U_g}\lmk A\rmk=
\lmk \cdots\otimes {U_g}\otimes {U_g}\otimes {U_g}\otimes \cdots\rmk
A\lmk \cdots\otimes {U_g}^{-1}\otimes {U_g}^{-1}\otimes {U_g}^{-1}\otimes \cdots\rmk,
\end{align}
for any 
$A\in \caA_\bbZ$ and  $g\in G$.

An interaction is a map $\Phi$ from 
${\mathfrak S}$ into ${\caA}_{{\mathbb Z}}^{\rm loc}$ such
that $\Phi(X) \in {\caA}_{X}$ 
and $\Phi(X) = \Phi(X)^*$
for $X \in {\mathfrak S}$. 
An interaction $\Phi$ is translation invariant if
$
\Phi(X+j)=\tau_j\lmk
\Phi(X)\rmk,
$
for all $ j\in{\mathbb Z}$ and $X\in  {\mathfrak S}$.
Furthermore, it is of finite range if there exists an $m\in {\mathbb N}$ such that
$
\Phi(X)=0$,
for $X$ with diameter larger than $m$. 
In this case, 
we say that the interaction length of $\Phi$ is less than or equal to $m$.
We denote the set of all translation invariant finite range interactions by $\caJ$. Furthermore, the set of all translation invariant interactions with interaction length less than or equal to $m$ is denoted by $\caJ_m$.
For the product action $\beta_U$ of a group $G$,
we say an interaction $\Phi$ is $\beta_U$-invariant if
$\beta_g\lmk \Phi(X)\rmk=\Phi(X)$ for all $X\in  {\mathfrak S}$.
A natural number $m\in\nan$ and an element $h\in\caA_{[0,m-1]}$,
define an interaction $\Phi_h$ by
\begin{align}\label{hamdef}
\Phi_h(X):=\left\{
\begin{gathered}
\tau_x\lmk h\rmk,\quad \text{if}\quad  X=[x,x+m-1] \quad \text{for some}\quad  x\in\bbZ\\
0,\quad\text{otherwise}
\end{gathered}\right.
\end{align}
for $X\in {\mathfrak S}$.
If $h$ is $\beta_U$-invariant, i.e., $\beta_{U_g}\lmk h\rmk=h$ for all $g\in G$,
then the interaction $\Phi_h$ is also $\beta_U$-invariant.
A Hamiltonian associated with $\Phi$ is a net of self-adjoint operators $H_{\Phi}:=\left((H_{\Phi })_\Lambda\right)_{\Lambda\in{\mathfrak I}}$ such that 
\begin{equation}\label{GenHamiltonian}
\lmk H_{\Phi}\rmk_{\Lambda}:=\sum_{X\subset{\Lambda}}\Phi(X).
\end{equation}
Note that $(H_\Phi)_{\Lambda}\in {\caA}_{\Lambda}$. 
Let us specify what we mean by gapped Hamiltonian in this paper.
\begin{defn}
A Hamiltonian $H:=\left(H_\Lambda\right)_{\Lambda\in{\mathfrak I}}$
associated with a positive translation invariant finite range interaction is \emph{gapped}  
if there exists $\gamma>0$ and $N_0\in\nan$ 
such that the difference 
between the smallest and the next-smallest eigenvalue of $H_\Lambda$, is bounded below by 
$\gamma$,  for all finite intervals $\Lambda\subset\bbZ$ with $|\Lambda|\ge N_0$. 
\end{defn}
In this definition, the smallest eigenvalue can be degenerated in general.

Now we introduce the $C^1$-classification of gapped Hamiltonians with $\beta_U$-symmetry.
We say  $\Phi:[0,1]\ni t\mapsto\Phi(t)\in  {\caJ}$ 
is a continuous and piecewise $C^1$-path if  for each $X\in{\mathfrak S}$, the path
$[0,1]\ni t\mapsto \Phi(t;X)\in {\caA}_X$ is continuous and piecewise $C^1$
with respect to the norm topology.
\begin{defn}[$C^1$-classification of gapped Hamiltonians with symmetry]\label{def:phasec}
Let $U:G\to \caU_n$ be a unitary representation of a group $G$ on $\bbC^n$.
Let $H_0,H_1$ be gapped  Hamiltonians associated with $\beta_U$-invariant interactions
$\Phi_{0},\Phi_{1}\in{\caJ}$.
We say that  $H_0,H_1$ are
$C^1$-equivalent with $\beta_U$-symmetry if the following conditions are satisfied.
\begin{enumerate}
\item There exists $m\in\nan$ and a continuous and piecewise $C^1$-path $\Phi:[0,1]\to {\caJ}_m$ such that $\Phi(0)=\Phi_0$, 
$\Phi(1)= \Phi_1$.
\item  Let   $H(t)$ be the Hamiltonian associated with $\Phi(t)$ for each $t\in[0,1]$.
There are $\gamma>0$, $N_0\in\nan$, and finite intervals $I(t)=[a(t), b(t)]$, whose endpoints $a(t), b(t)$ smoothly depending on $t\in[0,1]$,
such that for all finite intervals $\Lambda\subset\bbZ$ with $|\Lambda|\ge N_0$,
the smallest eigenvalue of $H(t)_\Lambda$ is in $I(t)$ and
the rest of the spectrum is in $[b(t)+\gamma,\infty)$.
\item
For each $t\in[0,1]$, $\Phi(t)$ is $\beta_U$-invariant.
\end{enumerate}
\end{defn}
\begin{rem}
We write $H_0\simeq_{C^1,U} H_1$ when
$H_0,H_1$ are 
$C^1$-equivalent with $\beta_U$-symmetry.
If furthermore the path  $\Phi(t)$ can be taken in ${\caJ}_m$ we write
$H_0\simeq_{C^1,U,m}H_1$.
\end{rem}

Now let us introduce the family of Hamiltonians we consider in this paper.
As an input, we prepare the following triple.
\begin{defn}\label{def:sp}
Let $G$ be a group and 
$c : G\times G\to \bbT$ a $2$-cocycle of $G$. Let $U:G\to \caU_n$ be a unitary representation of $G$ on $\bbC^n$.
Let $k$ be a natural number, $u_g$ a projective unitary representation of $G$ on $\bbC^k$ with respect to
the $2$-cocycle $c$,
and $V$ an isometry from $\bbC^k$ to $\bbC^n\otimes\bbC^k$.
We denote by $\caSP(n,G,U,c)$ the set of all such triple $(k, u, V)$ which satisfies the following conditions.
\begin{enumerate}
\item[(i)] For any $g\in G$, we have
\begin{align}\label{eq:uui}
\lmk {U_g}\otimes u_g\rmk V=V u_g.
\end{align}
\item[(ii)] Define $\vv=\vv_{(k,u,V)}=(v_1,\ldots, v_n)\in\mk^{\times n}$ by
\begin{align}\label{eq:Vv}
V x=\sum_{\mu=1}^{n}\psi_{\mu}\otimes v_{\mu}^* x,\quad
x\in\bbC^k.
\end{align}
(Recall that $\{\psi_{\mu}\}_{\mu}$ is the fixed CONS of $\cc^n$.)
Define the completely positive map $T_{\vv}\;:\; \mk\to\mk$ by
\begin{align}
T_{\vv}(X):=\sum_{\mu=1}^n v_{\mu}Xv_{\mu}^*,\quad X\in\mk.
\end{align}
Then $T_{\vv}$ is primitive.
\end{enumerate}
\end{defn}
The definition of primitivity can be found in \cite{Wolf:2012aa}, for example.

From the $n$-tuple of $k\times k$ matrices $\vv=\vv_{(k,u,V)}=(v_1,\ldots, v_n)\in\mk^{\times n}$
associated with $(k, u, V)\in \caSP(n,G,U,c)$, we construct our interaction following the recipe of
\cite{Fannes:1992vq}.
For $l\in\nan$ and $\mu^{(l)}=(\mu_0,\mu_1,\ldots,\mu_{l-1})\in\{1,\cdots,n\}^{\times l}$,
we use the notation
\begin{align}
\widehat{v_{\mu^{(l)}}}:=
v_{\mu_0}v_{\mu_1}\cdots v_{\mu_{l-1}}\in\mk,\quad
\widehat{\psi_{\mu^{(l)}}}:=
\bigotimes_{i=0}^{l-1}\psi_{\mu_i}\in\bigotimes_{i=0}^{l-1}\cc^n.
\end{align}
\begin{defn}
Let  $(k, u, V)\in \caSP(n,G,U,c)$ and $\vv=\vv_{(k,u,V)}=(v_1,\ldots, v_n)\in\mk^{\times n}$ the $n$-tuple
associated to it.
For each $l\in\nan$,
define
$\Gamma_{l,\vv}\;:\; \mk\to\bigotimes_{i=0}^{l-1}\cc^n$
by
\begin{align}\label{eq:gbdef}
\Gamma_{l,\vv}\lmk X\rmk
=\sum_{\mu^{(l)}\in \{1,\cdots,n\}^{\times l}}\lmk\Tr X \lmk\widehat{v_{\mu^{(l)}}} \rmk^*\rmk\widehat{\psi_{\mu^{(l)}}},\quad
X\in\mk,
\end{align}
and set
$
\cgv{l}:=\Ran\Gamma_{l,\vv}\subset \bigotimes_{i=0}^{l-1}\cc^n.
$
Furthermore, we denote by
$G_{l,\vv}$
the orthogonal projection onto $\cgv{l}$ in $\bigotimes_{i=0}^{l-1}\cc^n$.
We set $h_{l,\vv}:=1-G_{l,\vv}$ and $\Phi_{l,\vv}:=\Phi_{h_{l,\vv}}$.
(Recall (\ref{hamdef}).)
\end{defn}

It is known that this recipe gives a Hamiltonian with a spectral gap \cite{Fannes:1992vq}.
\begin{lem}\label{lem:spgap}
Let  $(k, u, V)\in \caSP(n,G,U,c)$ and $\vv=\vv_{(k,u,V)}=(v_1,\ldots, v_n)\in\mk^{\times n}$ the $n$-tuple
associated to it.
Then, for any $m\in\nan$, the interaction $\Phi_{m,\vv}$ is $\beta_U$-invariant.
For $m\ge k^4+1$ the Hamiltonian $H_{\Phi_{m,\vv}}$ is gapped.
\end{lem}
\begin{proof}
To show the first statement, let us consider a positive operator
\[
X_m:=\sum_{\mu^{(m)},\nu^{(m)}\in \{1,\cdots,n\}^{\times m}}
\lmk \Tr \lmk \lmk\widehat{v_{\mu^{(m)}}} \rmk^*\widehat{v_{\nu^{(m)}}} \rmk
\rmk
\ket{\widehat{\psi_{\mu^{(m)}}}}\bra{\widehat{\psi_{\nu^{(m)}}}}.
\]
We claim that the support $s(X_m)$ of $X_m$ is $G_{m,\vv}$.
To see $G_{m,\vv}\le s(X_m)$, let us consider an arbitrary $\xi\in \bigotimes_{i=0}^{m-1}\cc^n$, and
set $Z_{\xi}:=\sum_{\mu^{(m)}}\braket{\xi}{\widehat{\psi_{\mu^{(m)}}}}\lmk\widehat{v_{\mu^{(m)}}} \rmk^*\in\mk$.
For any $Y\in\mk$, we have $\braket{\xi}{\Gamma_{m,\vv}(Y)}=\Tr Y Z_\xi$.
We also have $\braket{\xi}{X\xi}=\Tr Z_{\xi} Z_{\xi}^*$.
Therefore, from the Cauchy-Schwartz inequality, we obtain
\[
\lv \braket{\xi}{\Gamma_{m,\vv}(Y)}\rv^2=\lv \Tr Y Z_\xi\rv^2
\le \Tr Y Y^*\cdot \Tr Z_\xi^*Z_\xi
=\Tr Y Y^*\cdot \braket{\xi}{X\xi},
\]
for any  $Y\in\mk$ and $\xi\in \bigotimes_{i=0}^{m-1}\cc^n$. This implies $G_{m,\vv}\le s(X_m)$.
To prove $G_{m,\vv}\ge s(X_m)$,
let $\{\ek{i,j}\}_{i,j=1}^k$ be the matrix units of $\mk$.
By the straightforward  calculation, we obtain
\[
X_m=\sum_{i,j=1}^k \ket{\Gamma_{m,\vv}\lmk \ek{i,j}\rmk}\bra{\Gamma_{m,\vv}\lmk \ek{i,j}\rmk}.
\]
This proves $G_{m,\vv}\ge s(X_m)$.

From this, we see that in order to show $G_{m,\vv}$ is $\beta_U$-invariant, it suffices to show that $X_m$ is $\beta_U$-invariant.
For this, we 
define a completely positive map $\bbE: \mn\otimes\mn\to\mk$ by
\[
\bbE\lmk A\otimes X\rmk := V^* \lmk A\otimes X\rmk V 
=\sum_{\mu,\nu=1}^{n} \braket{\psi_\mu}{A\psi_{\nu}} v_{\mu} X v_{\nu}^*
\quad
A\in \mn,\quad X\in\mk.
\]
For each $A\in\mn$, we set $\bbE_A:\mk\to\mk$ by $\bbE_A(X)=V^* \lmk A\otimes X\rmk V$, 
$X\in\mk$.
From (i) of the Definition \ref{def:sp}, we have
\begin{align}\label{eq:aug}
&\bbE_{\Ad\lmk  {U_g}^*\rmk\lmk A\rmk}\lmk X\rmk
= V^* \lmk {U_g}^*A {U_g}\otimes X\rmk V
=V^*\lmk {U_g}\otimes u_g\rmk^* \lmk A\otimes u_gXu_g^* \rmk
\lmk {U_g}\otimes u_g\rmk V\notag\\
&=\lmk V u_g\rmk^* \lmk A\otimes u_gXu_g^* \rmk
V u_g
=\Ad u_g^*\circ \bbE_{A}\circ \Ad u_g\lmk X\rmk.
\end{align}
(Here $\Ad \lmk {U_g}^*\rmk(X):={U_g}^* X {U_g}$ etc.)
On the other hand, by the straightforward calculation we obtain the following formula for any $A_0,\ldots, A_{m-1}\in\mn$,
\begin{align}\label{eq:xfm}
\Tr_{\bigotimes_{i=0}^{m-1}\mn}\lmk X_m
\bigotimes_{i=0}^{m-1} A_i\rmk
=\Tr_{\mk}\lmk \bbE_{A_0}\circ\cdots\circ\bbE_{A_{m-1}}\lmk 1\rmk\rmk.
\end{align}
Now we show the $\beta_U$-invariance of $X_m$.
For all $A_0,\ldots, A_{m-1}\in\mn$,
\begin{align*}
&\Tr_{\bigotimes_{i=0}^{m-1}\mn}\lmk \Ad\lmk  U_g^{\bigotimes m-1}\rmk \lmk X_m\rmk
\bigotimes_{i=0}^{m-1} A_i\rmk
=\Tr_{\bigotimes_{i=0}^{m-1}\mn}\lmk  X_m
\lmk \bigotimes_{i=0}^{m-1} \Ad \lmk U_g^*\rmk\lmk A_i\rmk\rmk \rmk\\
&=\Tr_{\mk}\lmk \bbE_{\Ad U_g^*\lmk A_0\rmk}\circ\cdots\circ\bbE_{\Ad U_g^*\lmk A_{m-1}\rmk}\lmk 1\rmk\rmk
=\Tr_{\mk}\lmk \bbE_{A_0}\circ\cdots\circ\bbE_{A_{m-1}}\lmk 1\rmk\rmk
=\Tr_{\bigotimes_{i=0}^{m-1}\mn}\lmk X_m
\bigotimes_{i=0}^{m-1} A_i\rmk.
\end{align*}
Here, in the second and the fourth equality, we used (\ref{eq:xfm}), and for the third equality, we used 
(\ref{eq:aug}).
This proves $\beta_{U_g}(X_m)=X_m$, for all $g\in G$.

The second statement follows from \cite{Fannes:1992vq}, except for the
condition $m\ge k^4+1$ which is discussed in Lemma 3.1 of \cite{bo}
and basically quantum {Wielandt's} inequality \cite{Sanz:2010aa}.
\end{proof}

\begin{defn}
Let $G$ be a group and 
$c : G\times G\to \bbT$ a $2$-cocycle of $G$.
Let $(u_0, \bbC^{k_0})$, $(u_1,\bbC^{k_1})$
be finite dimensional projective unitary representations of  $G$ with respect to the $2$-cocycle $c$.
We say $(u^{(0)}, \bbC^{k_0})$ and $(u^{(1)},\bbC^{k_1})$ are unitary equivalent if $k:=k_0=k_1$ and
there exists a unitary matrix $W\in\caU_k$ such that 
$W u^{(0)}_g=u^{(1)}_gW$, for all $g\in G$.
\end{defn}
Here is the main Theorem of this paper.
\begin{thm}\label{thm:gmain}
Let $G$ be a group and 
$c : G\times G\to \bbT$ a $2$-cocycle of $G$. Let $U:G\to \caU_n$ be a unitary representation of $G$ on $\bbC^n$.
Let $(k^{(0)}, u^{(0)}, V^{(0)})$ and $(k^{(1)}, u^{(1)}, V^{(1)})$ be elements in 
$\caSP(n,G,U,c)$, and $\vv_0:=\vv(k^{(0)}, u^{(0)}, V^{(0)})$,
 $\vv_1:=\vv(k^{(1)}, u^{(1)}, V^{(1)})$, the $n$-tuples associated to them via (\ref{eq:Vv}).
Assume that the projective representations $(u^{(0)}, \bbC^{k_0})$ and $(u^{(1)},\bbC^{k_1})$ are unitary equivalent.
Then we have $H_{\Phi_{m,\vv_0}}\simeq_{C^1,U,m}
H_{\Phi_{m,\vv_1}}$, for any $m\ge k^4+1$.
\end{thm}
\section{Proof of Theorem \ref{thm:gmain}}
We introduce an equivalence relation in $\caSP(n,G,U,c)$.
\begin{defn}\label{def:spath}
Let $(k^{(0)}, u^{(0)}, V^{(0)})$ and $(k^{(1)}, u^{(1)}, V^{(1)})$ be elements in 
$\caSP(n,G,U,c)$. We say $(k^{(0)}, u^{(0)}, V^{(0)})$ and $(k^{(1)}, u^{(1)}, V^{(1)})$ are equivalent,
 if the following holds.
\begin{enumerate}
\item$k:=k^{(0)}=k^{(1)}$.
\item There exists a map $u:[0,1]\times G\to \caU_{k}$ such that
for any $g\in G$, $[0,1]\ni t\to u(\cdot,g)\in \caU_{k}$ is $C^{\infty}$.
\item There exits a $C^{\infty}$-map $V:[0,1]\to B(\bbC^k, \bbC^{n}\otimes \bbC^{k})$
such that $(k,u(t,\cdot), V(t))\in\caSP(n,G,U,c)$, for all $t\in[0,1]$,
with $(k, u(0,\cdot), V(0))=(k^{(0)}, u^{(0)}, V^{(0)})$ and $(k, u(1,\cdot), V(1))=(k^{(1)}, u^{(1)}, V^{(1)})$.
\end{enumerate}
When $(k^{(0)}, u^{(0)}, V^{(0)})$ and $(k^{(1)}, u^{(1)}, V^{(1)})$ are equivalent, we write $(k^{(0)}, u^{(0)}, V^{(0)})\simeq_{\caSP}(k^{(1)}, u^{(1)}, V^{(1)})$.
\end{defn}
In order to prove Theorem \ref{thm:gmain}, we show 
$(k^{(0)}, u^{(0)}, V^{(0)})\simeq_{\caSP}(k^{(1)}, u^{(1)}, V^{(1)})$.
Note that if $u$ is a projective representation of $G$ with respect to a $2$-cocycle $c$, for any $W\in\caU_k$,
the map $G\ni g\mapsto W u_g W^*$ defines a a projective unitary representation of $G$ with respect to $c$.
We denote it by $\Ad W(u)$. 
\begin{lem}\label{lem:uch}
Let $(k,u,V)\in\caSP(n,G,U,c)$ and $W\in\caU_k$.
Then the triple $(k,\Ad W(u) , \lmk \unit\otimes W\rmk VW^*)$
belongs to $\caSP(n,G,U,c)$. 
In particular, we have
\[
(k,u,V)\simeq_{\caSP}(k,\Ad W(u) , \lmk \unit\otimes W\rmk VW^*).
\]
\end{lem}
\begin{proof}
It is clear that $\Ad W(u)$ is a projective unitary representation of $G$ with respect to a $2$-cocycle $c$, and that $\lmk \unit\otimes W\rmk VW^*$ is an isometry.
It is straightforward to check
\[
\lmk {U_g}\otimes \Ad W\lmk u_g\rmk\rmk \lmk \unit\otimes W\rmk VW^*=\lmk \unit\otimes W\rmk VW^* \Ad W\lmk u_g\rmk.
\]
Let $\vv$ be the $n$-tuple associated to $(k,u,V)$ via (\ref{eq:Vv}).
We then have 
\[
\lmk \unit\otimes W\rmk VW^* x=\sum_{\mu=1}^{n}\psi_{\mu}\otimes \lmk Wv_{\mu}W\rmk^* x,\quad
x\in\bbC^k.
\]
The primitivity of $T_\vv$ implies the primitivity of
$T_{(Wv_\mu W^*)_{\mu=1}^n}$. 
Therefore, we have $(k,\Ad W(u) , \lmk \unit\otimes W\rmk VW^*)\in \caSP(n,G,U,c)$.

To show the second statement, let $H\in \mk$ be a self adjoint element such that $W=e^{iH}$.
Then, the path $[0,1]\ni t\mapsto \hat W(t):= e^{itH}\in \caU_k$ is $C^{\infty}$, and the path of the triple
$[0,1]\ni t\mapsto \lmk k,\Ad \hat W(t)\lmk u \rmk, \lmk \unit\otimes \hat W(t)\rmk V\hat W(t)^* \rmk$ satisfies
 the conditions in Definition \ref{def:spath}.
To prove the third condition in Definition \ref{def:spath}, we use the first statement of this Lemma, which we have just proven.
\end{proof}
From Lemma \ref{lem:uch}, we see that in order to prove Theorem \ref{thm:gmain},
it suffices to consider the case $k^{(0)}=k^{(1)}$ and $u^{(0)}=u^{(1)}$.
Therefore, we prove the following Lemma.
\begin{lem}\label{lem:psp}
Let $(u, \bbC^{k})$
be finite dimensional projective unitary representation of $G$ with respect to the $2$-cocycle $c$.
Let $V_0,V_1:\bbC^k\to\bbC^n\otimes \bbC^k$ be isometries 
such that $(k,u, V_0),(k,u,V_1)\in\caSP(n,G,U,c)$.
Then we have 
\[
(k,u, V_0)\simeq_{\caSP}(k,u,V_1).
\] 
\end{lem}
In order to prove this, we first investigate the structure of the isometry $V$.
We consider all the equivalence classes of finite dimensional irreducible unitary $c$-projective representations of $G$.
Let $\left\{(\pi_\alpha, V_\alpha)\right\}_{\alpha}$ be the set of representatives of them.
Then we obtain the irreducible decompositions of $u_g$ and $U_g\otimes u_g$ \cite{Kiri}.
Note that the proof of Schur's Lemma and the irreducible decomposition for usual representations
works for unitary projective representations.
For each irreducible $c$-projective representation $\pi_\alpha$, there exist
numbers $m_\alpha,m_\alpha'\in\nan\cup\{0\}$. There also exist
unitaries $W:\bbC^k\to \bigoplus_{\alpha:m_\alpha\neq 0} V_{\alpha}\otimes \bbC^{m_\alpha}$, $W':\bbC^n\otimes \bbC^k\to \bigoplus_{\alpha:m_\alpha'\neq 0} V_{\alpha}\otimes \bbC^{m_\alpha'}$
and we have
\begin{align}
&Wu_g W^*=\bigoplus_{\alpha:m_\alpha\neq 0}\pi_\alpha(g)\otimes \unit_{m_\alpha},\label{eq:iru}\\
&W'\lmk U_g\otimes u_g\rmk {W'}^* =\bigoplus_{\alpha:m_\alpha'\neq 0}\pi_\alpha(g)\otimes \unit_{m_\alpha'}.
\label{eq:irUu}
\end{align}
We have the following Lemma.
\begin{lem}\label{lem:vst}
Let $(k,u,V)\in \caSP(n,G,U,c)$. We consider the irreducible decompositions given in (\ref{eq:iru}) and
(\ref{eq:irUu}).
Then $m_\alpha=0$ if $m_{\alpha}'=0$.
Furthermore, if $m_{\alpha}\neq 0$,
there exists an isometry $\omega_\alpha:\bbC^{m_\alpha}\to\bbC^{m_{\alpha}'}$ and they satisfy 
\[
V= {W'}^*\lmk\bigoplus_{\alpha:m_\alpha\neq 0}\unit_{V_\alpha}\otimes \omega_\alpha\rmk W.
\]
\end{lem}
\begin{proof}
For $m\in\nan$, let $\{\chi_i^{(m)}\}_{i=1}^m$ be the standard basis of $\bbC^{m}$.
Each element $\xi$ in $\bigoplus_{\alpha:m_\alpha\neq 0} V_{\alpha}\otimes \bbC^{m_\alpha}$
can be decomposed as
\[
\xi=\bigoplus_{\alpha:m_\alpha\neq 0} \sum_{i=1}^{m_\alpha}\xi_{\alpha,i}\otimes \chi_{i}^{(m_\alpha)},
\]
with some $\xi_{\alpha,i}\in V_\alpha$.
For $\alpha$ with $m_\alpha\neq 0$ and $i=1,\ldots, m_{\alpha}$, let us consider the subspace of $\bigoplus_{\beta:m_\beta\neq 0} V_{\beta}\otimes \bbC^{m_\beta}$
which consists of $\xi$ such that $\xi_{\beta,j}=0$ if $(\beta,j)\neq(\alpha,i)$.
We denote the orthogonal projection onto this subspace by $P_{\alpha,i}$.
Similarly, for $\alpha$ with $m_\alpha'\neq 0$ and $i=1,\ldots, m_{\alpha}'$ we define an orthogonal projection 
$P_{\alpha,i}'$
on $\bigoplus_{\alpha:m_\alpha'\neq 0} V_{\alpha}\otimes \bbC^{m_\alpha'}$.
Substituting the decompositions (\ref{eq:iru}) and
(\ref{eq:irUu}) to (\ref{eq:uui}), 
we obtain
\begin{align}
\lmk \bigoplus_{\alpha:m_\alpha'\neq 0}\pi_\alpha(g)\otimes \unit_{m_\alpha'}\rmk W'VW^*
=W'VW^*\lmk \bigoplus_{\alpha:m_\alpha\neq 0}\pi_\alpha(g)\otimes \unit_{m_\alpha}\rmk,\quad g\in G.
\end{align}
Multiplying $P_{\alpha,i}$ from the left and $P_{\beta,j}$ from the right of this equation,
we obtain
\[
\pi_\alpha(g)P_{\alpha,i}W'VW^*P_{\beta,j}
=P_{\alpha,i}W'VW^*P_{\beta,j}\pi_\beta(g),\quad g\in G,
\] 
for all $\alpha,\beta$ such that $m_{\alpha}'\neq 0$ and $m_\beta\neq 0$.
Regarding $P_{\alpha,i}W'VW^*P_{\beta,j}$ a linear map from $V_\beta$ to $V_\alpha$, we apply Schur's Lemma.
We then obtain $P_{\alpha,i}W'VW^*P_{\beta,j}=0$ if $\alpha\neq\beta$.
If $m_\alpha,m_\alpha'\neq 0$, there exists a scalar $C_{\alpha,i,j}\in\bbC$ such that
$P_{\alpha,i}W'VW^*P_{\alpha,j}=C_{\alpha,i,j}\unit_{V_\alpha}$.
Set $\omega_\alpha:=\sum_{\stackrel{1\le i\le m_\alpha}{1\le j\le m_{\alpha'}}}C_{\alpha,i,j}\ket{\chi_i^{(m_\alpha')}}\bra{\chi_j^{(m_\alpha)}} \in B(\bbC^{m_{\alpha}}, \bbC^{m_\alpha})$ in this case.
We then obtain
\begin{align*}
W'VW^*=\sum_{\stackrel{\alpha,\beta:m_\alpha\neq 0,m_\beta'\neq 0}{{1\le i\le m_\alpha, 1\le j\le m_{\beta}'}}}
P_{\alpha,i}W'VW^*P_{\beta,j}
=\sum_{\stackrel{\alpha:m_\alpha\neq 0,m_\alpha'\neq 0}{{1\le i\le m_\alpha, 1\le j\le m_{\alpha}'}}}P_{\alpha,i}W'VW^*P_{\alpha,j}
=\bigoplus_{{\alpha:m_\alpha\neq 0,m_\alpha'\neq 0}}\unit_{V_\alpha}\otimes \omega_\alpha.
\end{align*}
As $V$ is an isometry, we obtain
\begin{align*}
\bigoplus_{\alpha:m_\alpha\neq 0}\unit_{V_{\alpha}}\otimes \bbC^{m_\alpha}=\unit=\lmk W'VW^*\rmk^*W'VW^*=\bigoplus_{{\alpha:m_\alpha\neq 0,m_\alpha'\neq 0}}\unit_{V_\alpha}\otimes \omega_\alpha^*\omega_\alpha.
\end{align*}
If there exits an $\alpha$ such that $m_\alpha\neq0$ and $m_\alpha'= 0$, this equality can not hold. Therefore, we have
$m_\alpha'\neq0$ if $m_\alpha\neq 0$.
Furthermore, $\omega_\alpha:\bbC^{m_\alpha}\to\bbC^{m_{\alpha}'}$ is an isometry. 
\end{proof}
\begin{proofof}[Lemma \ref{lem:psp}]
For $U$ and $u$, we consider the irreducible decompositions (\ref{eq:iru}), (\ref{eq:irUu}).
Applying Lemma \ref{lem:vst} to $V_i$, $i=0,1$, we obtain isometries  
$\omega_{i,\alpha}:\bbC^{m_\alpha}\to \bbC^{m_\alpha'}$,  such that
\[
V_i= {W'}^*\lmk\bigoplus_{\alpha:m_\alpha\neq 0}\unit_{V_\alpha}\otimes \omega_{i,\alpha}\rmk W.
\]
For each $\alpha$ with $m_\alpha\neq 0$,
there exists a unitary $S_\alpha\in\caU_{m_\alpha'}$ such that $S_\alpha\omega_{0,\alpha}=\omega_{1,\alpha}$.
This is because $\omega_{0,\alpha}$, $\omega_{1,\alpha}$ are isometries.
Let $H_\alpha\in\Mat_{m_{\alpha}'}$ be a self-adjoint matrix such that $S_\alpha=e^{iH_\alpha}$.
We set $H_\alpha=0$ if $m_\alpha'\neq0$ and $m_\alpha=0$. We define $H:=\bigoplus_{\alpha:m_\alpha'\neq 0}\unit_{V_\alpha}\otimes H_\alpha$.

We would like to connect $(k,u, V_0)$ and $(k,u,V_1)$ by a smooth path in $\caSP(n,G,U,c)$, by connecting
$\unit$ and $e^{iH}$ suitably.
In order for that, we recall the  the necessary and sufficient condition for the primitivity, introduced in
\cite{SW}. See Appendix \ref{sec:prim}. We use the notations in Appendix \ref{sec:prim}.
We define a $B(\bbC^k,\bbC^n\otimes\bbC^k)$-valued entire analytic function $V(z)$ by
\begin{align}\label{eq:vz}
V(z):={W'}^*e^{izH}\lmk\bigoplus_{\alpha:m_\alpha\neq 0}\unit_{V_\alpha}\otimes \omega_{0,\alpha}\rmk W,\quad
z\in\bbC.
\end{align}
Note that $V(0)=V_0$ and $V(1)=V_1$.From the definition (\ref{eq:vz}), and the decompositions (\ref{eq:iru}), (\ref{eq:irUu}), we obtain
\begin{align}\label{eq:vg}
\lmk {U_g}\otimes u_g\rmk V(z)=V(z) u_g,\quad z\in\bbC, \quad g\in G.
\end{align}
As in (\ref{eq:Vv}), we define $v_{\mu}\lmk  z\rmk$
\begin{align*}
V(z) x=\sum_{\mu=1}^{n}\psi_{\mu}\otimes \lmk v_{\mu}\lmk \bar z\rmk\rmk^*  x,\quad
x\in\bbC^k,\quad z\in\bbC.
\end{align*}
Note that $\bbC\ni z\mapsto v_{\mu}\lmk z\rmk$ is entire analytic, for each $\mu=1,\ldots,n$.
We write $\vv(z):=(v_1(z),\ldots, v_n(z))$.
By the same calculation as in (\ref{eq:aug}), we obtain
\begin{align}\label{eq:utg}
\Ad u_g^*\circ T_{\vv(z)}\circ \Ad u_g=T_{\vv(z)},\quad z\in\bbC, \quad g\in G.
\end{align}
For each $\mm:=(\mu^{(k^4)}_i)_{i=1}^{k^2}$, where $\mu^{(k^4)}_i\in\left\{1,\ldots,n\right\}^{\times k^4}$, $i=1,\ldots,k^2$, we define an entire analytic function
\[
f_\mm(z):=\braket{\zeta}{\bigotimes_{i=1}^{k^2}\lmk\widehat { v_{\mu_i^{(d^4)}}(z)}\otimes \unit\rmk \Omega}.
\]
We set
\[
\caZ:=\bigcap_{\mm}\left\{z\in\bbC\mid 
f_\mm(z)=0
\right\}.
\]
Here the intersection is taken over all $\mm:=(\mu^{(k^4)}_i)_{i=1}^{k^2}$, where $\mu^{(k^4)}_i\in\left\{1,\ldots,n\right\}^{\times k^4}$, $i=1,\ldots,k^2$.
From Lemma \ref{lem:pg}, $T_{\vv(z)}$ is primitive if and only if there exists $\mm$ such that $f_\mm(z)\neq0$.
Therefore, $T_{\vv(z)}$ is primitive if and only if $z\notin\caZ$. In particular, we have
$0,1\notin \caZ$.
This means at least one of $f_\mm$ is not identically zero.
As each $f_\mm(z)$ is entire analytic and at least one of them is not identically zero, the intersection of $\caZ$ and a ball $\{z\in \cc\mid |z|\le 2\}$
is a finite set.

Let $\tilde \varpi:[0,1]\to \bbC$ be a path in $\bbC$ given by $\tilde \varpi(t)=t$, $t\in[0,1]$.
As $\caZ\cap\{z\in \cc\mid |z|\le 2\}$ is a finite set, we can deform this $\tilde \varpi$ and obtain a $C^\infty$-path
$\varpi:[0,1]\to \bbC$ with $\varpi(0)=0$, $\varpi(1)=1$ such that $\varpi(t)\notin \caZ$ for all $t\in[0,1]$.
In particular, $T_{\vv(\varpi(t))}$ is primitive for $t\in[0,1]$.
From this primitivity, the spectral radius $r_{T_{\vv(\varpi(t))}}$ of
 $T_{\vv(\varpi(t))}$ is strictly positive, 
and it is a non-degenerate eigenvalue of $T_{\vv(\varpi(t))}$.
Let $P_{\{r_{T_{\vv(\varpi(t))}}\}}^{T_{\vv(\varpi(t))}}$ be the spectral projection 
of $T_{\vv(\varpi(t))}$ onto $\{r_{T_{\vv(\varpi(t))}}\}$. 
Then $e_{\vv(\varpi(t))}:=P_{\{r_{T_{\vv(\varpi(t))}}\}}^{T_{\vv(\varpi(t))}}\lmk \unit_k\rmk$ is a strictly positive element in
$\mk$
and there exists a faithful state $\varphi_{\vv(\varpi(t))}$ on $\mk$ such that
$P_{\{r_{T_{\vv(\varpi(t))}}\}}^{T_{\vv(\varpi(t))}}(X)= \varphi_{\vv(\varpi(t))}\lmk X\rmk e_{\vv(\varpi(t))}$, for
$X\in \mk$.
By (\ref{eq:utg}), $e_{\vv(\varpi(t))}$ and $u_g$ commute for all $t\in[0,1]$ and $g\in G$.
Note that $r_{T_{\vv(\varpi(0))}}=r_{T_{\vv(\varpi(1))}}=1$ and
$ e_{\vv(\varpi(0))}= e_{\vv(\varpi(1))}=1$.
Furthermore, the maps $[0,1]\ni t\mapsto r_{T_{\vv(\varpi(t))}}, e_{\vv(\varpi(t))}, \varphi_{\vv(\varpi(t))}$
are $C^\infty$.
Hence, setting 
\[
\hat v_\mu(t):= r_{T_{\vv(\varpi(t))}}^{-\frac 12}  e_{\vv(\varpi(t))}^{-\frac 12}
v_\mu(\varpi(t)) e_{\vv(\varpi(t))}^{\frac 12},\quad t\in[0,1],\quad \mu=1,\ldots,n,
\]
we obtain $C^\infty$ maps in $\mk$.
We set $\hat \vv(t):=(\hat v_1(t),\ldots, \hat v_n(t))$, $t\in[0,1]$.
Note that $\hat\vv(0)=\vv(\varpi(0))=\vv(0)$ and $\hat\vv(1)=\vv(\varpi(1))=\vv(1)$.
By this definition, $T_{\hat \vv(t)}$ is a unital completely positive map.
As $T_{\vv(\varpi(t))}$ is primitive, this $T_{\hat \vv(t)}$  is also primitive.

We define a $C^\infty$-path $\hat V: [0,1]\to B(\bbC^k,\bbC^n\otimes \bbC^k)$,
by
\[
\hat V(t)x:=\sum_{\mu=1}^n \psi_\mu\otimes \hat v_\mu(t)^*x
,\quad x\in \bbC^k,\quad
t\in[0,1].
\]
Note that $\hat V(0)=V(0)=V_0$, $\hat V(1)=V(1)=V_1$.
In order to prove Lemma \ref{lem:psp}, it suffices to show 
$(k,u,\hat V(t))\in \caSP(n,G,U,c)$, for all $t\in[0,1]$.
Since $T_{\hat \vv(t)}$ is unital, $\hat V(t)$ is an isometry.
As we already know that $T_{\hat \vv(t)}$
is primitive, what we have to check is (i)  of Definition \ref{def:sp}.
This can be checked as follows.
For any $x\in \bbC^k$,
\begin{align*}
&\lmk {U_g}\otimes u_g\rmk \hat V(t)x
=\lmk {U_g}\otimes u_g\rmk \lmk \sum_{\mu=1}^n \psi_\mu\otimes \hat v_\mu(t)^*x\rmk
= r_{T_{\vv(\varpi(t))}}^{-\frac 12}  \lmk
\sum_{\mu=1}^n  {U_g}\psi_\mu\otimes  u_g e_{\vv(\varpi(t))}^{\frac 12}
\lmk v_\mu(\varpi(t))\rmk^* e_{\vv(\varpi(t))}^{-\frac 12}
x\rmk
\\
&=r_{T_{\vv(\varpi(t))}}^{-\frac 12}  \lmk
\sum_{\mu=1}^n  {U_g}\psi_\mu\otimes  e_{\vv(\varpi(t))}^{\frac 12}u_g
\lmk v_\mu(\varpi(t))\rmk^* e_{\vv(\varpi(t))}^{-\frac 12}
x\rmk\\
&=r_{T_{\vv(\varpi(t))}}^{-\frac 12}  
\lmk \unit_n\otimes e_{\vv(\varpi(t))}^{\frac 12}\rmk
\lmk {U_g}\otimes u_g\rmk
\lmk
\sum_{\mu=1}^n  \psi_\mu\otimes  
\lmk v_\mu(\varpi(t))\rmk^* e_{\vv(\varpi(t))}^{-\frac 12}
x\rmk\\
&=r_{T_{\vv(\varpi(t))}}^{-\frac 12}  
\lmk \unit_n\otimes e_{\vv(\varpi(t))}^{\frac 12}\rmk
\lmk {U_g}\otimes u_g\rmk
V\lmk \overline{\varpi(t)}\rmk e_{\vv(\varpi(t))}^{-\frac 12}
x
=r_{T_{\vv(\varpi(t))}}^{-\frac 12}  
\lmk \unit_n\otimes e_{\vv(\varpi(t))}^{\frac 12}\rmk
V\lmk \overline{\varpi(t)}\rmk u_g e_{\vv(\varpi(t))}^{-\frac 12}
x
\\
&=r_{T_{\vv(\varpi(t))}}^{-\frac 12}  
\lmk \unit_n\otimes e_{\vv(\varpi(t))}^{\frac 12}\rmk
V\lmk \overline{\varpi(t)}\rmk e_{\vv(\varpi(t))}^{-\frac 12} u_g
x
=r_{T_{\vv(\varpi(t))}}^{-\frac 12}  
\lmk \unit_n\otimes e_{\vv(\varpi(t))}^{\frac 12}\rmk
\lmk
\sum_{\mu=1}^n  \psi_\mu\otimes  \lmk v_\mu(\varpi(t))\rmk^*e_{\vv(\varpi(t))}^{-\frac 12} u_g
x\rmk\\
&=
\lmk
\sum_{\mu=1}^n  \psi_\mu\otimes  r_{T_{\vv(\varpi(t))}}^{-\frac 12} e_{\vv(\varpi(t))}^{\frac 12}\lmk v_\mu(\varpi(t))\rmk^*e_{\vv(\varpi(t))}^{-\frac 12} u_g
x\rmk
=
\sum_{\mu=1}^n \psi_\mu\otimes \hat v_\mu(t)^*u_gx
=\hat V(t) u_gx,\quad g\in G,\quad t\in[0,1].
\end{align*}
In the third and the seventh equality, we used the commutativity of $ e_{\vv(\varpi(t))}$ and $u_g$.
In the sixth equality, we used (\ref{eq:vg}).
This completes the proof.
\end{proofof}
\begin{proofof}[Theorem \ref{thm:gmain}]
As $(u^{(0)}, \bbC^{k^{(0)}})$ and $(u^{(1)},\bbC^{k^{(1)}})$ are unitary equivalent, we have $k:=k^{(0)}=k^{(1)}$ and
there exists a unitary matrix $W\in\caU_k$ such that 
$W u^{(0)}_g=u^{(1)}_gW$, $g\in G$.
By Lemma \ref{lem:uch}, we have
$(k^{(0)}, u^{(0)}, V^{(0)})\simeq_{\caSP}(k^{(0)},\Ad W(u^{(0)}) , \lmk \unit\otimes W\rmk V^{(0)}W^*)
=(k,u^{(1)}, \lmk \unit\otimes W\rmk V^{(0)}W^*)$.
Furthermore, by Lemma \ref{lem:psp}, we obtain $(k,u^{(1)}, \lmk \unit\otimes W\rmk V^{(0)}W^*)\simeq_{\caSP}
(k,u^{(1)},V_1)=(k^{(1)}, u^{(1)},V_1)$.
Applying Lemma \ref{lem:spgap}, the same argument as \cite{bo} implies Theorem \ref{thm:gmain}.
\end{proofof}
\appendix

\section{Notations}\label{sec:nota}
For  $k\in\nan$, the set of all $k\times k$ matrices
over $\cc$ is denoted by
$\mk$.Furthermore, we denote the set of unitary elements of $\mk$
by $\caU_k$. For $A\in\mk$, we denote the map $\mk\ni X\mapsto AXA^*\in\mk $
by $\Ad A$.
For a linear map $\Gamma$, $\Ker\Gamma$, and $\Ran\Gamma$ denote the kernel and the range of $\Gamma$ respectively.
For a finite dimensional Hilbert space, braket $\braket{}{}$ denotes the
inner product of
the space under consideration.
We denote the set of all bounded linear maps from Hilbert space
$\mathcal H$ to $\mathcal K$ is denoted by
$B({\mathcal H},{\mathcal K})$.
For  a subspace $\mathfrak K$, ${\mathfrak{K}}^{\perp}$ means the orthogonal complement of $\mathcal H$ .

\section{Primitivity}\label{sec:prim}
In this section we recall the necessary and sufficient condition for $T_\vv$, (given for an $n$-tuple of $k\times k$ matrices $\vv\in\mk^{\times n}$), to be primitive. This condition was introduced in \cite{SW}. See \cite{SW} for the detail.Let $\{\chi_i^{(k)}\}_{i=1}^k$ (resp. $\{\chi_i^{(k^2)}\}_{i=1}^{k^2}$)  be a standard basis of $\bbC^k$ (resp. $\bbC^{k^2}$).
We set $\Omega:=\sum_{i=1}^k\chi_i^{(k)}\otimes \chi_i^{(k)}$.
We also set $\zeta:=\sum_{\sigma\in S_{k^2}}\sgn\sigma \chi_{\sigma(1)}^{(k^2)}\otimes\chi_{\sigma(2)}^{(k^2)}\otimes \cdots
\otimes \chi_{\sigma(k^2)}^{(k^2)}$.
Here, $S_{k^2}$ is the symmetric group of degree $k^2$ and $\sgn\sigma$ is the signature of $\sigma\in S_{k^2}$. 
\begin{lem}[\cite{SW}]\label{lem:pg}
The completely positive map $T_\vv$ is primitive if and only if
there exists $\mm:=(\mu^{(k^4)}_i)_{i=1}^{k^2}$, where $\mu^{(k^4)}_i\in\left\{1,\ldots,n\right\}^{\times k^4}$, $i=1,\ldots,k^2$ such that
\[
\braket{\zeta}{\bigotimes_{i=1}^{k^2}\lmk v_{\mu_i^{(d^4)}}\otimes \unit\rmk \Omega}\neq 0.
\]
\end{lem}

\end{document}